\documentclass[10pt,final,twocolumn]{IEEEtran}

\usepackage[latin1]{inputenc}
\usepackage[english]{babel}
\usepackage{amsmath,amssymb}
\usepackage{graphicx,graphicx,verbatim}
\usepackage{verbatim}
\usepackage{color}

\newtheorem{thm}{Theorem}
\newtheorem{conj}{Conjecture}
\newtheorem{cor}{Corollary}

\newtheorem{defi}{Definition}
\newtheorem{prop}{Proposition}

\newtheorem{exx}{Example}
\newtheorem{remm}{Remark}

\newenvironment{corollary}{\begin{cor}\rm }{\hfill \hspace*{1pt} \hfill $\lrcorner$ \end{cor}}
\newenvironment{proposition}{\begin{prop}\rm }{\hfill \hspace*{1pt} \hfill $\lrcorner$ \end{prop}}
\newenvironment{theorem}{\begin{thm}\rm }{\hfill \hspace*{1pt} \hfill $\lrcorner$ \end{thm}}
\newenvironment{definition}{\begin{defi}\rm }{\hfill \hspace*{1pt} \hfill $\lrcorner$ \end{defi}}
\newenvironment{remark}{\begin{remm}\rm }{\hfill \hspace*{1pt} \hfill $\lrcorner$\end{remm}}

\newenvironment{proofof}{\noindent {\em Proof of }}{\hfill \hspace*{1pt}
\hfill $\blacksquare$}
\newenvironment{proof}{\noindent {\em Proof.}}{\hfill \hspace*{1pt} \hfill $\square$}

\newcommand\real{\ensuremath{{\mathbb R}}}
\newcommand\realn{\ensuremath{{\mathbb{R}^n}}}
\newcommand\mymatrix[2]{\left[\begin{array}{#1} #2 \end{array}\right]}
\newcommand{\smallmat}[1]{\left[ \begin{smallmatrix}#1
    \end{smallmatrix} \right]}

\newcommand{\calA}{\mathcal{A}}

\newcommand{\calH}{\mathcal{H}}

\newcommand{\calU}{\mathcal{U}}
\newcommand{\calV}{\mathcal{V}}

\newcommand{\calK}{\mathcal{K}}
\newcommand{\calW}{\mathcal{W}}

\newcommand{\calX}{\mathcal{X}}

\definecolor{dark-blue}{RGB}{30,30,160}
\definecolor{dark-green}{RGB}{30,160,30}

\begin{document}

\title{Differential dissipativity theory \\ for dominance analysis
\thanks{
F. Forni and R. Sepulchre are with the University of Cambridge, Department of Engineering, 
Trumpington Street, Cambridge CB2 1PZ \texttt{f.forni@eng.cam.ac.uk|r.sepulchre@eng.cam.ac.uk}.
The research leading to these results has received funding from the European Research Council under the
Advanced ERC Grant Agreement Switchlet n.670645.}}
\author{Fulvio Forni, Rodolphe Sepulchre}
\date{\today}

\maketitle

\begin{abstract}  % Abstract of not more than 200 words.
High-dimensional systems that have a low-dimensional {\it dominant} behavior allow for model reduction and simplified analysis. We use differential analysis to formalize this important concept in a nonlinear setting. We show that dominance can be studied through linear dissipation inequalities and an interconnection theory that closely mimics the classical analysis of stability by means of dissipativity theory. In this approach, stability is seen as the particular situation where the dominant behavior is $0$-dimensional. The generalization opens novel tractable avenues to study multistability through $1$-dominance and limit cycle oscillations through $2$-dominance.
\end{abstract}

\section{Introduction}

The analysis of a system is considerably simplified when it is low-dimensional. Linear system analysis frequently exploits the property that a few {\it dominant} poles capture the main properties of a possibly high-dimensional system. Low-dimensional models are even more critical in nonlinear system analysis. Multistability or limit cycle analysis is difficult beyond the phase plane analysis of two-dimensional systems.

In this paper we seek to formalize the property that a nonlinear system has low-dimensional {\it dominant} behavior. Our approach is {\it differential}: we characterize the property for linear systems and then study the nonlinear system differentially, that is, along the linearized flow in the tangent bundle.  The seminal example of differential analysis in control theory is contraction analysis \cite{Lohmiller1998,Pavlov2005,Fromion2005,Russo2010,Sontag2010}, which we interpret as a differential analysis of exponential stability. The property is the contraction of a ball,  characterized via a Lyapunov dissipation inequality.  A nonlinear system is contractive when this dissipation inequality holds infinitesimally along any of its trajectories. In the present paper, the corresponding property is $0$-dominance: it ensures that the dominant behavior of the nonlinear system is $0$-dimensional. A more recent example of differential analysis is differential positivity \cite{Forni2016}, the differential analysis of positivity.  The linear property is the contraction of a cone, also characterized by a dissipation inequality. A nonlinear system is differentially positive when this dissipation inequality holds infinitesimally along any of its trajectories. In the context of the present paper, the corresponding property is $1$-dominance: it guarantees that the dominant behavior of the nonlinear system is $1$-dimensional. More generally, we study the property of $p$-dominance differentially. The property is characterized via a dissipation inequality, which is then required to hold infinitesimally along the trajectories of the nonlinear system. We prove a general theorem that formalizes that the dominant behavior of a $p$-dominant system is $p$-dimensional.

Differential analysis is general in that it allows to extend a linear system property to an arbitrary flow defined on a differentiable manifold. An important restriction in this paper is that we only  consider dissipation properties characterized by {\it linear} matrix  inequalities. Furthermore, we impose the linear dissipation inequality to be {\it uniform} on the tangent bundle. In geometric terms, this endows the differentiable manifold with a {\it flat} non-degenerate metric structure, characterized by a {\it constant} quadratic form.  

The linear-quadratic restriction allows for a fruitful  bridge with the linear-quadratic theory of dissipativity.  Many of the available computational tools of dissipativity theory become available for the analysis of $p$-dominance. The analysis of $p$-dominance is reformulated as the search of a quadratic differential {\it storage} that decreases along the solutions of the linearized dynamics. The classical interconnections theorems of (linear-quadratic) dissipativity theory are reformulated as tools that facilitate that construction. The only substantial difference with the classical theory is that the quadratic form that characterizes the storage is no longer required to be a Lyapunov function, i.e. positive definite. Instead, it is required to have a fixed inertia, that is, $p$ negative eigenvalues and $n-p$ positive eigenvalues. Stability corresponds to $p=0$, whereas $p$-dominance allows for any integer $0 \le p \le n$. 

The paper is organized as follows. Section \ref{sec:p-dominance} characterizes $p$-dominance and $p$-dissipativity for linear time-invariant models. 
Section  \ref{sec:diffdominance} defines $p$-dominance for nonlinear systems with two main results that characterize their asymptotic behavior. The connection with related results
in the literature is discussed in Section \ref{sec:literature_comparison} and dominance analysis via the solution of LMIs is illustrated in Section \ref{sec:p-dominanceLMI}. 
Differential $p$-dissipativity is addressed in Section \ref{sec:interconnections}, where we primarily illustrate the role of  differential passivity and differential small-gain theorems
in the design and analysis of multistable systems ($p=1$) and limit cycle oscillations ($p=2$).
The proofs of the main theorems are provided in appendix.

\section{Dominant LTI systems}
\label{sec:p-dominance}

\begin{definition}
\label{def:LMI-dominance}
A linear system $\dot{x} = Ax$ is  $p$-dominant with rate $\lambda \ge 0$ 
if there exists a symmetric matrix $P$ with inertia $(p,0,n-p)$  such that
\begin{equation}
\label{eq:LMI-dominance}
A^T P + P A \le  -2 \lambda P - \varepsilon I \ . 
\end{equation}
for some $\varepsilon \ge 0$. The property is {\it strict} if $\varepsilon >0$.
\end{definition}
We recall that a matrix with inertia $(p,0,n-p)$ has $p$ negative
eigenvalues, and $n-p$ positive eigenvalues. For simplicity in what follows
we will use the abbreviated terminology \emph{inertia $p$} to denote those matrices.

In terms of the quadratic form $V(x)=x^TPx$, the (strict) dissipation inequality (\ref{eq:LMI-dominance}) reads
\begin{eqnarray*}
\dot{V}(x) & =  & x^T (A^T P + P A) x \\
& \le &  -2 \lambda V(x)  - \varepsilon \mid x \mid^2
\end{eqnarray*}
For $\varepsilon > 0$
this implies that the two cones
$$ \calK^- = \{ x \in  \real^n \mid V(x) \le 0 \} , \;  \calK^+ = \{ x \in  \real^n \mid V(x) \ge 0 \}$$
are strictly contracting either in forward or in backward time:
$$
 \forall  t > 0:   e^{-At}  \calK^+ \subset   \calK^+ ,   e^{At}  \calK^- \subset   \calK^-    
$$

Equivalent characterizations of $p$-dominance are provided in the following proposition,
whose proof is in the preliminary version of this paper \cite{Forni2017a}.
\begin{proposition}
\label{prop:equivalences}
For $\varepsilon>0$, the Linear Matrix Inequality \eqref{eq:LMI-dominance} is equivalent to any of the following conditions:
\begin{enumerate}
\item The matrix $A+ \lambda I $ has $p$ eigenvalues with strictly positive real part
and $n-p$ eigenvalues with strictly negative real part.
\item there exists an invariant {\it subspace} splitting  $\real^n=  \cal H \oplus \cal V$ such that $A \cal H \subset \cal H$ and $A \cal V \subset \cal V$. The dimension of $\cal V$ is $n-p$ and  the dimension of $\calH$ is $p$.
Furthermore, there exist
constants  $0 < \underline{C} \leq 1 \leq \overline{C}$ and $\underline{\lambda} < \lambda < \overline{\lambda}$ such that
\begin{eqnarray*}
\forall x \in {\cal H}: \;  \mid e^{At} x \mid  \geq   \underline{C} \,e^{-\underline{\lambda} t} \mid x \mid,  \; t \ge 0    \\
\forall x \in {\cal V}: \;  \mid e^{At} x \mid  \leq   \overline{C}\, e^{-\overline{\lambda} t} \mid x \mid, \; t \ge 0.
\end{eqnarray*}
\end{enumerate}
\end{proposition}

The property of $p$-dominance ensures a splitting between $n-p$ {\it transient} modes
and $p$ {\it dominant} modes.  Only the $p$ dominant modes dictate the asymptotic behavior.
Because $\lambda \ge 0$ and $\calV \subset \calK^+$, the quadratic form $V(x)$ is a Lyapunov function for the transient behavior, that is, for  the restriction of the flow in  $\calV$.

For $p=0$, $p$-dominance is the classical property of exponential stability: all modes are transient and the asymptotic behavior
is $0$-dimensional. 

The matrix inequality \eqref{eq:LMI-dominance} is equivalent to the conic constraint
\begin{equation}
\label{eq:internal}
\mymatrix{c}{\!\!\dot x \!\! \\ \!\! x \!\!}^T \!
\mymatrix{cc}{
0 & P \\ P & 2\lambda P + \varepsilon I
}
\mymatrix{c}{\!\!\dot x \!\! \\ \!\! x \!\!}
\leq 0.
\end{equation}

Dissipativity theory  extends $p$-dominance to {\it open} systems
by augmenting the internal dissipation inequality with an external supply.
The external property of  \emph{$p$-dissipativity} is captured by a
conic constraint between the state of the system $x$, its derivative $\dot{x}$,
and the external variables $y$ and $u$ of the form
\begin{equation}
\label{eq:p-dissipativity}
\mymatrix{c}{\!\!\!\dot x\!\!\! \\ \!\!\!x\!\!\!}^T \!\!
\mymatrix{cc}{
0 & P \\ P & 2\lambda P + \varepsilon I
}\!
\mymatrix{c}{\!\!\!\dot x\!\!\! \\ \!\!\!x\!\!\!}
\leq
\mymatrix{c}{\!\!\!y\!\!\! \\ \!\!\!u\!\!\!}^T \!\!
\mymatrix{cc}{
\!Q\! & \!L\! \\ \!L^T\! & \!R\!
}\!
\mymatrix{c}{\!\!\!y\!\!\! \\ \!\!\!u\!\!\!} 
\end{equation}
where $P$ is a matrix with inertia $p$, $\lambda \ge 0 $,
$L, Q,R$ are matrices of suitable dimension,
and $\varepsilon \geq 0$. The property is {\it strict} if $\varepsilon >0$.
We call supply rate
$s(y,u) := y^T Q y + y^TL u + u^T L^T y + u^T R u$ the right-hand side of \eqref{eq:p-dissipativity}.
An open dynamical system  
is \emph{$p$-dissipative with rate $\lambda$} if its
dynamics $\dot{x} = Ax + Bu$, $y=Cx + Du$ satisfy \eqref{eq:p-dissipativity}
for all $x$ and $u$. 
The property has a simple characterization in terms of 
linear matrix inequalities.
\begin{proposition}
\label{prop:LMI-dissipavitity}
A linear system $\dot{x} = Ax + Bu$, $y=Cx + Du$ is $p$-dissipative with rate $\lambda \ge 0$
if and only if there exists a 
symmetric matrix $P$ with inertia $p$ such that
\begin{equation}
\label{eq:LMI-dissipativity}
\smallmat
{\!
A^T\!  P + P A + 2\lambda P -C^T\! Q C + \varepsilon I \!&\! P B - C^T\!L - C^T\! Q D \! \\
 \! B^T P - L^T C - D^T Q C \!&\! -D^T Q D - L^TD - D^T L - R \!
 }
 \leq 0\, .
\end{equation}
\end{proposition}
\begin{proof} $ $
[$\Rightarrow$]
Just replace $\dot{x} = Ax+Bu$ and $y = Cx+Du$ in \eqref{eq:p-dissipativity}
and rearrange. [$\Leftarrow$] Multiply \eqref{eq:LMI-dissipativity} by
$[x^T \, u^T]$ on the left, and by $[x^T \, u^T ]^T$ on the right. Then we get
$
\dot{x}^T P x + x^T P \dot{x} +2\lambda x^T Px <  s(y,u)
$
 as desired.
\end{proof}

An interconnection theorem can be easily derived. A proof is provided in the preliminary version 
of this paper \cite{Forni2017a}, see also the proof of Theorem \ref{thm:interconnection}.
\begin{proposition}
\label{thm:dissipativity}
Let $\Sigma_1$ and $\Sigma_2$ 
$p_1$-dissipative and $p_2$-dissipative
systems respectively, with
uniform rate $\lambda$ and 
with supply rate
\begin{equation}
s_i(y_i,u_i) = \mymatrix{c}{\!\!y_i\!\! \\ \!\!u_i\!\!}^T \!
\mymatrix{cc}{
Q_i & L_i \\ L_i^T & R_i
}
\mymatrix{c}{\!\!y_i\!\! \\ \!\!u_i\!\!} 
\end{equation}
for $i \in \{1,2\}$. The closed-loop system given 
by negative feedback interconnection
\begin{equation}
\label{eq:feedback_interconnection}
u_1 = -y_2 + v_1 \qquad u_2 = y_1 + v_2
\end{equation}
is $(p_1+p_2)$-dissipative with rate $\lambda$
from $v = (v_1,v_2)$ to $y = (y_1,y_2)$
with supply rate 
\begin{equation}
\label{eq:cl_supply}
s(y,v) = 
\mymatrix{c}{\!\!\!y \!\!\! \\ \!\!\!v\!\!\!}^{\!T} \!\!
{\footnotesize 
\mymatrix{cc|cc}{
Q_1 + R_2 & -L_1 + L_2^T & L_1 & R_2 \\
-L_1^T + L_2 & Q_2  + R_1 & -R_1 & L_2 \\ \hline
L_1^T &  -R_1 & R_1 & 0 \\
R_2 & L_2^T & 0 & R_2 
}\!
}\!
\mymatrix{c}{\!\!\!y\!\!\! \\ \!\!\!v\!\!\!}  .
\end{equation}
Furthermore, the closed-loop system is $(p_1+p_2)$-dominant with rate $\lambda$ if 
\begin{equation}
\label{eq:dissipativity_interconnection}
\mymatrix{cc}{
Q_1 + R_2 & -L_1 + L_2^T \\ -L_1^T + L_2 & Q_2  + R_1
} 
\leq 0 \ . \vspace{-3mm}
\end{equation}
\end{proposition}

Mimicking classical dissipativity theory, there are two important particular cases of supply
rates : the passivity supply
\begin{equation}
\label{eq:passive}
s(y,u) = \mymatrix{c}{\!\!y\!\! \\ \!\!u\!\!}^T \!
\mymatrix{cc}{
0 & I \\ I & 0
}
\mymatrix{c}{\!\!y\!\! \\ \!\!u\!\!}  \ .
\end{equation}
and the gain supply:
\begin{equation}
\label{eq:L2}
s(y,u) = \mymatrix{c}{\!\!y\!\! \\ \!\!u\!\!}^T \!
\mymatrix{cc}{
-I & 0 \\ 0 & \gamma^2 I
}
\mymatrix{c}{\!\!y\!\! \\ \!\!u\!\!}  \ .
\end{equation}
Hence, Proposition  \ref{thm:dissipativity} provides an analog of the small-gain theorems and passivity theorems for 
$p$-dominance of a linear time-invariant system.

\section{Differential analysis of \\ dominant nonlinear systems}

\label{sec:diffdominance}

For a nonlinear system 
\begin{equation}
\label{eq:system}
\dot{x} = f(x)  \qquad x\in \calX
\end{equation}
we define dominance {\it differentially}, that is, through the linear dissipation inequality
\begin{equation}
\label{eq:diff-internal}
\mymatrix{c}{\!\!\dot {\delta x}\!\! \\ \!\!\delta x\!\!}^T \!
\mymatrix{cc}{
0 & P \\ P & 2\lambda P + \varepsilon I
}
\mymatrix{c}{\!\!\dot {\delta x}\!\! \\ \!\!\delta x\!\!}
\leq 0
\end{equation}
for every $\delta x \in T_x\calX$, 
where $P$ is a matrix with inertia $p$
and  $\varepsilon \ge 0$.

In what follows we assume that $\calX$ is a smooth Riemannian  manifold of dimension $n$.
Given any $\delta x \in T_x\calX$, $| \delta x |$ denotes the Riemannian metric on $\calX$ 
represented by $\sqrt{\delta x^T\! \delta x}$ in local coordinates.
$\psi^t(x)$ denotes the flow of \eqref{eq:system} at time $t$ passing through $x\in \calX$ at time $0$.
$\partial \psi^t(x)$ denotes the differential of $\psi^t(x)$ with respect to $x$.  Note that $(\psi^t(x), \partial \psi^t(x)\delta x)$ is a flow in the tangent bundle.
It is the solution of the prolonged system 
\cite{Crouch1987}
\begin{equation}
\label{eq:prolonged}
\left\{
\begin{array}{rcl}
\dot{x} &=& f(x) \\
\dot{\delta x} &=& \partial f(x) \delta x 
\end{array}
\right.
\qquad (x,\delta x)\in T\calX
\end{equation}
at time $t$ passing through  $(x,\delta x)\in T\calX$ at time $0$.
In local coordinates $\partial f(x)$ denotes the differential of $f$ at $x$
and $P$ is the local representation of a  metric tensor $P$  with
fixed inertia. 

Following the approach of Section \ref{sec:p-dominance},
we define the property of dominance for nonlinear systems
as follows.
\begin{definition}
\label{def:diff_dominance}
A nonlinear system $\dot{x} = f(x)$ is  $p$-dominant with rate $\lambda \ge 0$ 
if there exists a symmetric matrix $P$ with inertia $p$  such that
\eqref{eq:diff-internal} is satisfied by the solutions of the prolonged system \eqref{eq:prolonged}
for some $\varepsilon \ge 0$. The property is {\it strict} if $\varepsilon >0$.
\end{definition}

In terms of the quadratic function $V(\delta x) := \delta x^T P \delta x$, the differential dissipation inequality (\ref{eq:diff-internal}) reads
\begin{eqnarray}
\label{eq:useful_ineq}
\nonumber \dot{V} (\delta x) & = & \delta x^T\! \left( \partial f(x)^T \! P + P \partial f(x)  \right) \delta x  \\
\label{diffdissipation} & \le &  - 2 \lambda V(\delta x)   - \varepsilon \mid \delta x \mid^2
\end{eqnarray}
As a consequence, for $\varepsilon>0$, the two cone fields 
 \begin{eqnarray*}
  \calK^+(x) := \{\delta x\in T_x \calX \,|\,V(\delta x) \geq 0\},  \\   \calK^-(x) := \{\delta x\in T_x \calX \,|\,V(\delta x) \leq 0\} 
  \end{eqnarray*} 
are strictly contracting either  in forward time or in backward time,
respectively: 
 \begin{eqnarray*}
 \partial \psi^{-t}(x)  \calK^+(x) \subset   \calK^+(\psi^{-t}(x) )  \qquad  \forall  t > 0 \\
 \partial \psi^t(x)   \calK^-(x)  \subset   \calK^-(\psi^t(x))   \qquad \forall  t > 0      
\end{eqnarray*} 

The following result provides the differential analog of Proposition \ref{prop:equivalences}.
\begin{theorem}
\label{thm:splitting}
Let $\calA \subseteq\calX$ be a compact invariant set and 
let  \eqref{eq:system} be a strictly $p$-dominant system with rate $\lambda \geq 0$ and tensor $P$.
Then,  for each $x\in\calA$, there exists an invariant splitting $T_x\calX = \calH_{x} \oplus \calV_{x}$
such that
\begin{equation}
\label{eq:splitting}
\begin{array}{rcll}
\partial \psi^t(x) \calH_{x} \!\! & \!\! \subseteq \!\! & \!\! \calH_{\psi^t(x) } &\quad \forall t \in \real \, ,  \vspace{1mm}\\
\partial \psi^t(x) \calV_{x} \!\! & \!\! \subseteq \!\! & \!\! \calV_{\psi^t(x) }  &\quad \forall t \in \real \, .
\end{array}
\end{equation}
 $\calH_x$ and $\calV_x$ are distributions of dimension $p$ and $n-p$, respectively. Furthermore,
 there exist constants $\underline{C} \leq 1 \leq \overline{C}$ and $\underline{\lambda} < \lambda < \overline{\lambda}$ such that
\begin{subequations}
\label{eq:domination}
\begin{align}
|\partial \psi^t(x) \delta x| & \geq \underline{C} e^{-\underline{\lambda} t} |\delta x| \quad \forall x \in \calA, \forall \delta x \in \calH_x \, \label{eq:dominationH} \\
|\partial \psi^t(x) \delta x| & \leq \overline{C} e^{-\overline{\lambda} t} |\delta x| \quad \forall x \in \calA, \forall \delta x \in \calV_x \, . \label{eq:dominationV} 
\end{align}
\end{subequations}
\end{theorem}

The interpretation of the theorem is that the linearized flow $\partial \psi^t(\cdot)$ admits an invariant splitting between $n-p$ transient modes
and $p$ dominant modes. The  $p$ dominant modes dictate the long-term behavior of the flow. 
The quadratic form $V(\delta x)$ is a 
differential Lyapunov function in the invariant distribution $\calV \subset T\calX$.

For $\calX = \realn$ the characterization of the asymptotic behavior of dominant
systems can be further refined. 
This is because, by integration, the differential dissipation inequality \eqref{eq:useful_ineq}
leads to the {\it incremental} inequality
\begin{equation}
\label{eq:incremental}
\begin{array}{rcl}
\!\dot{V}(x\!-\!y) \!\!&\!\!\!=\!\!&\!\! (x\!-\!y)^T\!P(f(x) \!-\! f(y)) \!+\! (f(x) \!-\! f(y))^T\!P(x\!-\!y) \\
\!\!&\!\!=\!\!&\!\! 2 (x-y) \left( \int_0^1 P \partial f(s x + (1\!-\! s) y) d s \right) (x-y) \\
\!\!&\!\!\leq \!\! & \!\! -\varepsilon (x-y)^T \left( \int_0^1 -2\lambda P + \varepsilon I \, ds \right)  (x-y)  \\
\!\!&\!\!\leq \!\! & \!\! - 2 \lambda V(x-y)   - \varepsilon \mid x-y \mid^2 \ .
\end{array}
\end{equation} 
The following result is based on the incremental dissipation inequality (\ref{eq:incremental}).  
We denote by $\Omega(x)$ the $\omega$-limit set of $x$, that is, the set of all $\omega$-limit points of $x$.

\begin{theorem}
\label{thm:reduced}
For $\calX = \realn$, let  \eqref{eq:system} be a 
strictly $p$-dominant system with rate $\lambda \geq 0$.
Then, the flow on any compact $\omega$-limit set is topologically equivalent 
to a flow on a compact invariant set of a Lipschitz system in $\real^p$.
\end{theorem}

For small values of $p$, Theorem \ref{thm:reduced} severely constrains the possible attractors of the system.

\begin{corollary}
\label{thm:asymptotic_behavior}
Under the assumptions of Theorem \ref{thm:reduced},
every bounded solution asymptotically converges to

$\bullet$~a unique fixed point if $p = 0$;

$\bullet$~a fixed point if $p = 1$;

$\bullet$~a simple attractor if $p=2$, that is, a fixed point, a set of fixed points and connecting arcs, or a limit cycle. 
\end{corollary}

\section{Proofs of Theorems \ref{thm:splitting} and \ref{thm:reduced}}

\begin{proofof}{\emph{Theorem \ref{thm:splitting}}.}

\underline{Invariant splitting.}
For any $p > 0$ and for  any $\delta x \in \calK^-$, the dissipation inequality (\ref{diffdissipation}) implies 
\begin{equation}
\label{eq:invarianceK-}
\dot V( \delta x) \le  -2 \lambda V(\delta x) - \varepsilon \mid \delta x \mid^2 \le - \varepsilon \mid \delta x \mid^2
\end{equation}
for all $x \in \cal X$ and all $\delta x$ on the boundary of $\calK^-$, 
which guarantees that 
\begin{subequations}
\label{eq:invariance2K-}
\begin{align}
\forall t \ge 0 : & \,\, \partial \psi^t \calK^- \subseteq \calK^-\\
\forall t > 0 : & \,\, \partial \psi^t (\calK^-\setminus\{0\}) \subset \calK^-.
\end{align}
\end{subequations}
The dissipation inequality (\ref{diffdissipation}) also implies
\begin{equation}
\label{expestimateK-}
  \dot V( \delta x) \le  -2 \lambda V(\delta x) - \varepsilon \mid \delta x \mid^2 \le - (2 \lambda - \varepsilon_1) V(\delta x) 
  \end{equation}
for $\varepsilon_1 := \frac{\varepsilon}{\mid \lambda_{\min}(P) \mid}  >0$. Time-integration of this inequality yields the estimate
\begin{equation}
\label{expestimate2K-}
 \forall t \ge 0 : \frac{e^{2 \lambda t} V (\partial \psi^t \delta x)}{V(\delta x)} \ge e^{\varepsilon_1 t} 
  \end{equation}
which holds uniformly for all $x \in \cal X$ and all  $\delta x$ in the interior of $\calK^-$. 
\eqref{eq:invariance2K-} and \eqref{expestimate2K-} guarantee that
there exist $T>0$ and $\mu > 1$ such that
$\frac{|e^{\lambda t} \partial \psi^t \delta x|}{|\delta x|} \geq \mu$ for all $t\geq T$,
all $x\in \calX$ and all $\delta x \in \calK^-$.

Likewise, for any $n-p > 0$ and for any $\delta x \in \calK^+$, the dissipation inequality (\ref{diffdissipation}) implies 
\begin{equation}
\label{expestimateK+}
\dot V( \delta x) \le  - 2 \lambda V(\delta x) - \varepsilon \mid \delta x \mid^2 \le -(2 \lambda + \varepsilon_2) V(\delta x) 
 \end{equation}
for  $\varepsilon_2 = \frac{\varepsilon}{\lambda_{\max}(P) }  >0$. 
Integration of the first inequality backward time 
guarantees that 
\begin{subequations}
\label{eq:invariance2K+}
\begin{align}
\forall t \ge 0 : & \,\, \partial \psi^{-t} \calK^+ \subseteq \calK^+\\
\forall t > 0 : & \,\, \partial \psi^{-t} (\calK^+\setminus\{0\}) \subset \calK^+.
\end{align}
\end{subequations}
Integration of the second inequality backward time also yields the estimate
\begin{equation}
\label{expestimate2K+}
\forall t \ge 0 : \frac{e^{-2 \lambda t} V (\partial \psi^{-t} \delta x)}{V(\delta x)} \ge e^{\varepsilon_2 t} 
\end{equation}
which holds uniformly for all $x \in \cal X$ and all  $\delta x$ in the interior of $\calK^+$.
As above, \eqref{eq:invariance2K+} and \eqref{expestimate2K+} guarantee that
there exist $T>0$ and $\mu > 1$ such that
$\frac{|e^{-\lambda t} \partial \psi^{-t} \delta x|}{|\delta x|} \geq \mu$ for all $t\geq T$,
all $x\in \calX$ and all $\delta x \in \calK^+$.

From here, we proceed as in the proof of \cite[Theorem 1.2]{Newhouse2004} 
(see also \cite[Chapter 3]{HandbookDynamicalSystemsV1A}) to show that
$$
\begin{array}{rcl}
\calH_x &:=& \bigcap_{t \geq 0}e^{\lambda t} \partial \psi^t(x) \calK^-(\psi^{-t}(x)) \subset \calK^-(x) \vspace{1mm}\\
\calV_x &:=& \bigcap_{t \geq 0}e^{-\lambda t} \partial \psi^{-t}(\psi^t(x)) \calK^+(\psi^{t}(x)) \subset \calK^+(x)
\end{array}
$$
are invariant distributions  of dimension $p$ and $n-p$ respectively, that is,
$$
\begin{array}{rcll}
e^{\lambda t} \partial \psi^t(x) \calH_{x} \!\! & \!\! \subseteq \!\! & \!\! \calH_{\psi^t(x) } &\quad \forall t \in \real \ ,  \vspace{1mm}\\
e^{\lambda t}\partial \psi^t(x) \calV_{x} \!\! & \!\! \subseteq \!\! & \!\! \calV_{\psi^t(x) }  &\quad \forall t \in \real \ .
\end{array}
$$
Since $e^{\lambda t}$ is just a scalar factor, \eqref{eq:splitting} follows.

\underline{Exponential estimates.} Observe that $\delta x \in \calH$ implies that $\delta x$ belongs to the interior of $\calK^-$.
The estimate \eqref{eq:dominationH} with $\underline{\lambda} = \lambda - \frac{\varepsilon_1}{2}$ follows from the fact that  
$-V(\delta x)$ is positive definite in $\calH$ and that $V(\delta x)$ satisfies (\ref{expestimateK-}). For instance, 
there exist $0 < \rho_1\leq \rho_2$ such that 
$\rho_1 \delta x^T\! \delta x \leq -V(\delta x) \leq \rho_2 \delta x^T\! \delta x $ for all $\delta x \in \calH$
and \eqref{expestimateK-} gives
$
\rho_2 |\partial \psi^t \delta x|^2 \geq -V (\partial \psi^t \delta x) \geq -e^{-(2 \lambda - \varepsilon_1) t} V(\delta x) \geq e^{-(2 \lambda - \varepsilon_1) t} \rho_1 | \delta x|^2
$
for all $t \geq 0$, from which \eqref{eq:dominationH} follows.

Likewise, 
$\delta x \in \calV$ implies that $\delta x$ belongs to the interior of $\calK^+$. 
The estimate \eqref{eq:dominationV} with  $\overline{\lambda} = \lambda + \frac{\varepsilon_2}{2}$ follows from the fact that  
$V(\delta x)$ is positive definite in $\calV$ and satisfies (\ref{expestimateK+}).
\end{proofof}
\vspace{2mm}

\begin{proofof}{\emph{Theorem \ref{thm:reduced}}.}
From the dissipation inequality (\ref{eq:incremental}) we derive the inequality
$
\frac{d}{dt} e^{2\lambda t} V(x-y) = e^{2\lambda t} \dot{V}
+ 2\lambda e^{2\lambda t} V(x-y) \leq - e^{2\lambda t} \varepsilon \! \mid x-y \mid^2 
$
which, by time integration, implies  the following estimate
for any pair of solutions initialized at $x_0,y_0\in\calX$: 
\begin{eqnarray}
\nonumber 
\forall t \ge 0: V(\psi^t(x_0)-\psi^t(y_0))  \leq   e^{-2\lambda t} V(x_0-y_0) - \\
\label{eq:quadratic_decay}
      - \varepsilon \!\! \int_0^t \!\!\!e^{2\lambda (\tau-t)} |\psi^\tau\!(x_0)-\psi^\tau\!(y_0)|^2 d\tau
\end{eqnarray}
For large $t\ge 0$, the first term on the right hand side vanishes. This implies that the difference between any two solutions either asymptotically vanishes or eventually remains in the cone $\calK^-$.
We conclude that if $x$ and $y$ are distinct $\omega$-limit points, then necessarily $V(x-y) < 0$. 

Consider any  compact set $\Omega$ of $\omega$-limit points. Let $\calH_P$ and $\calV_P$ the invariant subspaces of the matrix $P$ associated to the $p$ negative and $n-p$ positive eigenvalues, respectively.
Define the linear  projection $\Pi:\calX \to \calH_P$ parallel to $\calV_P$.
We claim that  $\Pi$ restricted to $\Omega$ 
is one-to-one. This is because $x \neq y $ and $\Pi (x-y) = 0$ imply  $V(x-y) > 0$, which was proved to contradict  (\ref{eq:quadratic_decay}).

The remaining argument follows the proof of \cite[Theorem 3.17]{Hirsch2006}.
If $ y \in \Pi\Omega(x)$ then $y = \Pi z$ for a unique $z \in \Omega(x)$ and the flow
$\Pi \psi^t(z)$ on $\calH_P$ is generated by the vector field
$$F(y) := \Pi f(\Pi^{-1}(y))\qquad y \in \Omega(x) \ ,$$
which is Lipschitz by construction.
\end{proofof}

\section{Connections with the literature}
\label{sec:literature_comparison}

\subsection{Dominated splittings and dominance}

The property of {\it dominance} studied in this paper is closely related to the cousin concepts of {\it dominated splittings} and
{\it partial hyperbolicity}. Both concepts have appeared in dynamical systems theory as part of the extensive research to generalize the key concept of
 hyperbolicity pioneered by Smale and Anosov in the 60's. The common theme of that research line is that robust features of smooth dynamical systems should
be captured by robust features of their linear approximations. Robust is to be understood here in the sense of structural stability, that is, robustness to small perturbations of the vector field.

The subject is vast but we refer the interested reader to the recent survey \cite{Sambarino2014} for an orientation map.  Both dominated splittings and partial hyperbolicity continue to be an important subject in dynamical systems theory \cite{Crovisier2015, Sambarino2014,Plante1972,Hirsch1977,Katok1997,Pesin2004}. We refer the reader to \cite{wiggins1994,Pesin2004} for the implications of Theorem \ref{thm:splitting} on normal hyperbolicity and structural stability of compact attractors. Dominated splittings have also received attention in control theory \cite{Colonius2012}.   

Theorem \ref{thm:splitting} and its proof are grounded in the  results and proofs of \cite[Theorem 1.2]{Newhouse2004} and \cite[Theorem 3]{Forni2016}. 
Theorem \ref{thm:reduced} and its proof are grounded in the  results and proofs of \cite[Theorem 3.17]{Hirsch2006} and \cite[Proposition 3]{Sanchez2009}.
The proof merges these approaches with the techniques in \cite{Smith1980}.

The decomposition $T\calX = \calH \oplus \calV$ in Theorem \ref{thm:splitting} is a splitting that is  invariant under the linearized flow $\partial \psi^t$.  The splitting is called {\it dominated} because the flow satisfies
$$ \frac{ |\partial \psi^t(x) \delta x_v |}{|\partial \psi^t(x) \delta x_h|} \le \frac{\overline{C}}{\underline{C}} e^{-(\overline{\lambda}-\underline{\lambda}) t} \frac{ |\delta x_v |}{| \delta x_h|} $$
for any $\delta x_h \neq 0 \in \calH_x$ and $\delta x_v \in \calV_x$.
The dominated splitting is a consequence of the contraction of the cone fields $\calK^+(x)$ and $\calK^-(x)$. 

The contraction of a cone is {\it projective}, that is, it expresses a contraction of the non-dominant directions {\it relative} to the dominant directions of the flow. Because of the assumption of a {\it nonnegative} dissipation rate $\lambda \ge 0$, dominance further imposes  vertical contraction, that is, contraction of the flow  $\partial \psi^t$ {\it in} the vertical distribution $\calV$, see (\ref{eq:dominationV}).  The requirement of vertical contraction 
is an extra requirement of dominance with respect to the property of dominated splitting.  Theorem \ref{thm:reduced} does not hold without this extra requirement. 

\subsection{Contraction, differential stability, and $0$-dominance}

A strict $0$-dominant system  is a contractive system  \cite{Lohmiller1998,Pavlov2005,Russo2010,Forni2014}. For a linear system,
the property is simply exponential stability, meaning hyperbolicity {\it and} contraction of the $n$ transient modes  to the $0$-dimensional attractor. Because $P$ is
positive definite, the dissipation inequality implies the contraction of an ellipsoid.
The quadratic form $V(\delta x)$ is  a differential Lyapunov function in the terminology of \cite{Forni2014}.
On vector spaces $\calX = \realn$, its integration along geodesic 
curves leads to the incremental Lyapunov function $V(x-y)$.
From (\ref{eq:incremental}), it implies
exponential contraction of the difference between any two trajectories. The attractor of a $0$-dominant system
is necessarily a unique fixed point.

\subsection{Monotonicity, differential positivity, and $1$-dominance}

Strictly $1$-dominant systems are
strictly differentially positive systems \cite{Forni2016,Forni2015}.  
The contractive cone $\calK^-$ is an ellipsoidal cone, that is, it is the union of two  solid pointed convex cones 
$\calK^- = -\calK^* \cup \calK^*$as illustrated in Figure \ref{fig:cone}.

\begin{figure}[htbp]
\centering
\includegraphics[width=0.5\columnwidth]{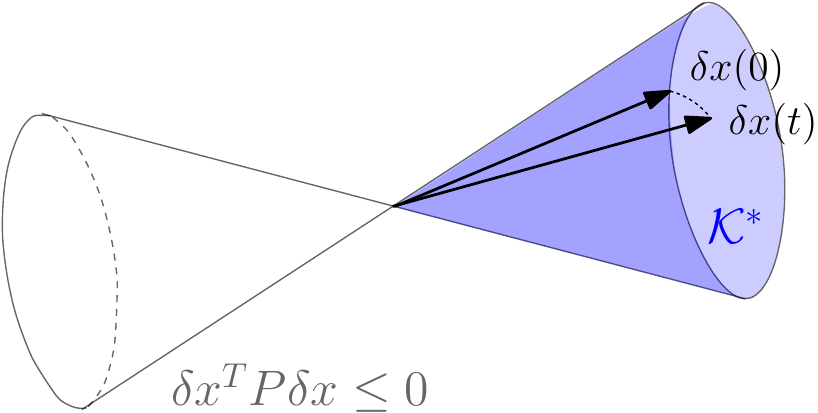}
\caption{For $\varepsilon>0$
\eqref{eq:diff-internal} guarantees that any trajectory moves from the boundary
of $\calK^*$ towards the interior.}
\label{fig:cone}
\end{figure}
A linear system that contracts a solid pointed convex cone is  a {\it positive} system \cite{Luenberger1979,Bushell1973}.
Positivity with respect to an ellipsoidal cone was characterized via LMIs in \cite{Stern1991,Stern1991a}.
See also \cite{Grussler2014} for a recent use of such characterization in model reduction.
Differential positivity thus induces a dominated splitting between one dominant direction and $n-1$ dominated directions.
The reader is referred to  \cite{Forni2016,Forni2015} for a comprehensive analysis of the asymptotic behavior of differentially positive systems. 

The distinction between $1$-dominance and differential positivity is again the property of {\it projective} contraction versus {\it vertical} contraction. Differential
positivity does not imply contraction in the $1$-dimensional vertical subspace $\calV_x$. As a consequence, the statement of  Corollary \ref{thm:asymptotic_behavior} for $p=1$ only holds for {\it generic} initial conditions
of strictly differentially positive systems, see \cite[Corollary 5]{Forni2016}.

For $\calX = \realn$, $1$-dominance is also tightly related to monotonicity 
\cite{Smith1995,Hirsch2003,Angeli2003,Hirsch2006}. 
A monotone system preserves a partial order $\preceq$: any pair of trajectories $x(\cdot)$, $y(\cdot)$ 
from ordered initial conditions
$x(0) \preceq y(0)$ satisfy
$x(t) \preceq y(t)$ for all $t\geq 0$. 

Monotonicity is implied by $1$-dominance. The partial order is the usual partial order associated to a pointed
convex cone: $x \preceq y \mbox{ iff } y-x \in \calK^*$. Monotonicity requires that trajectories $y(0) - x(0) \in \calK^*$ satisfy 
$y(t) - x(t) \in \calK^*$ for all $t\geq 0$, which is guaranteed by $1$-dominance by 
the invariance $V(x(t)-y(t)) \leq 0$ for all $t\geq 0$ and by the fact that if $x(0)-y(0) = 0$ then
$x(t)-y(t) = 0$, for all $t \geq 0$. Here also, monotonicity is independent of $\lambda$. The assumption $\lambda \ge 0$ makes  $1$-dominance stronger than monotonicity and the statement of  Corollary \ref{thm:asymptotic_behavior} for $p=1$ only holds for {\it generic} initial conditions
of monotone systems, see \cite{Hirsch1988,Hirsch2006}.

\subsection{Contraction of rank $2$ cones and $2$-dominance}

The property of $2$-dominance provides the following generalization of Poincar{\'e}-Bendixson theorem:
\begin{corollary}
\label{thm:poincare_bendixson}
For $p=2$,  
under the assumptions of Theorem \ref{thm:reduced}, 
let $\calU\subseteq\calX$ be a 
compact forward invariant set that does not contain fixed points.
Then, the $\omega$-limit set of any point in $\calU$ is a closed orbit. 
\end{corollary}
\begin{proof}
Take any $\omega$-limit set contained in $\calU$. 
By Theorem \ref{thm:reduced} the flow restricted to this set is topologically equivalent
to the flow of a planar system. By Poincar{\'e}-Bendixson theorem \cite[Chapter 11, Section 4]{Hirsch1974},
a nonempty compact limit set of a planar system which contains no fixed points is a closed orbit.
\end{proof}
A similar generalization was developed in the papers \cite{Smith1980,Smith1986,Sanchez2009} by generalizing
the concept of monotonicity to rank-$2$ cones. It is this generalization that motivated the results in the present paper.
Note that the statement of  Corollary \ref{thm:poincare_bendixson} only holds for generic initial conditions
of rank 2 monotone systems. The stronger conclusion of Corollary \ref{thm:poincare_bendixson} is again due to
the assumption of a nonnegative dissipation rate $\lambda \ge 0$.

\subsection{Invariant cone fields}

The assumption of a {\it constant} matrix $P$ makes the cone fields $\calK^+(x)$ and $\calK^-(x)$ {\it constant}, that is, the same
cone is attached to every $x \in \calX$. A more intrinsic characterization  for $\calX = \realn$ is that the cone field is {\it invariant} by translation,
the natural group action on a  vector space. This geometric interpretation allows for extensions on 
Lie groups and, more generally, homogeneous spaces \cite{Mostajeran2017,Mostajeran2018a,Bonnabel2009}.
The definition of dominance thus assumes an {\it invariant} cone field in the present paper.
The  invariance of the cone field is an important source of tractability for the search of the storage.

\section{Algorithmic test for $p$-dominance \\ and a simple example}
\label{sec:p-dominanceLMI}

\subsection{Dominant spectral splitting}

A necessary condition for dominance is that the spectrum of the family 
of  matrices $\partial f(x)+\lambda I$, $x \in \cal X$, admits a uniform splitting, 
see Figures \ref{fig:locus1} (left) and \ref{fig:locus2} (left) for an illustration.
\begin{theorem}
\label{thm:spectrum}
Let \eqref{eq:system} be a strictly $p$-dominant system. Then,
there exists a maximal interval $(\lambda_{\mathrm{min}}, \lambda_{\mathrm{max}})$
such that $\partial f(x)+\lambda I$  has $p$ unstable 
eigenvalues and $n-p$ stable eigenvalues (negative real part)
for each $\lambda\in(\lambda_{\mathrm{min}}, \lambda_{\mathrm{max}})$, for every $x \in \calX$.
\end{theorem}
 \begin{proof}
By Proposition \ref{prop:equivalences},
the feasibility of 
$\partial f(x)^TP + P\partial f(x) + 2\lambda P \leq -\varepsilon I$, for $\varepsilon > 0$
and for $P$ of inertia $p$ guarantees that each matrix $\partial f(x) + \lambda I$ 
has $p$ unstable eigenvalues and $n-p$ stable eigenvalues (negative real part)
at each $x\in \calX$. The splitting of the eigenvalues of $\partial f(x)$ at each $x$ 
is preserved for every value of $\lambda$ within some given spectral gap $(\lambda_{\mathrm{min}}(x), \lambda_{\mathrm{max}}(x))$.
Thus, by the uniformity of the strict inequality above,
$\lambda_{\mathrm{min}} := \sup_{x\in\calX} \lambda_{\mathrm{min}}(x) < \lambda$ and 
 $\lambda_{\mathrm{max}} := \inf_{x\in\calX} \lambda_{\mathrm{max}}(x) > \lambda$.
 \end{proof}
Spectral analysis of the Jacobian matrix $\partial f(x)$ is thus useful to select $p$ and $\lambda$
in dominance analysis. The uniform splitting of the spectrum is necessary but of course not sufficient
for dominance. This is well-known even for $p=0$. For instance, classical
counterexamples to Kalman's conjecture \cite{Kalman1957} illustrate that a system
can fail to be contractive even when the spectrum of its Jacobian is uniformly in the left half complex plane.

\subsection{Convex relaxations}
\label{sec:convex_relaxation}

Sufficient conditions for  dominance are provided by the inequality
\begin{equation}
\label{eq:diff-internalLMI}
\partial f(x)^T P + P \partial f(x) + 2\lambda P +\varepsilon I \leq 0\qquad \forall x \in \calX
\end{equation}
whose solutions $P$, for some $\varepsilon\geq 0$,
must also satisfy a fixed inertia constraint. 

If the spectrum of $\partial f(x)$  admits a stable splitting for a  given
$\lambda$, then  all the solutions $P$ of
\eqref{eq:diff-internalLMI} must share the same inertia, meaning that the inertia condition
can be dropped. One is then left with solving an infinite family of LMIs.

It is common practice to reduce an infinite family of LMIS to a finite family through
convex relaxation, see e.g. \cite{lmibook} and references therein.
Let $\calA := \{A_1, \dots, A_N\}$ be a family of matrices such that
$\partial f(x) \in \mathit{ConvexHull}(\calA)$ for all $x$. Then,
by construction, 
any (uniform) solution $P$ to 
\begin{equation}
\label{eq:diff-internalLMIrelax}
A_i^T P + P A_i + 2\lambda P + \varepsilon I \leq 0  \qquad 1\leq i\leq N
\end{equation}
is a solution to \eqref{eq:diff-internalLMI}. For instance, at each
$x$, $\partial f(x) = \sum_{i=1}^N \rho_i(x) A_i$ for a given set of $\rho_i(x)$
such that $\sum_{i=1}^N \rho_i(x) = 1$. Thus, the left-hand side of \eqref{eq:diff-internalLMI}
reads 
$
\left( \sum_{i=1}^N \rho_i(x) A_i^T\right) P + P \left( \sum_{i=1}^N \rho_i(x) A_i \right) + 2\lambda P + \varepsilon I
= 
\sum_{i=1}^N \rho_i(x) \left( A_i^T P + P A_i  + 2\lambda P + \varepsilon I \right)
\leq 0,
$
where the last inequality follows from \eqref{eq:diff-internalLMIrelax}.

The algorithmic steps of dominance analysis of a given nonlinear system $\dot x = f(x)$ can thus
be summarized as follows:
\begin{enumerate}
\item Estimate $p$ and $\lambda$ from the spectrum analysis of $\partial f(x)$
\item Reduce the infinite family of LMIs to a finite family by convex relaxations.
\item Test the feasibility of the relaxed LMI with a LMI solver.
\end{enumerate}

\subsection{Example}
\label{sec:examples_dominance}

We illustrate the theory on a classical textbook example: a one degree of freedom mechanical
system with nonlinear spring (Duffing model), actuated by a DC motor with a PI feedback control.
While this example is elementary, it illustrates the tractability of dominance analysis on a four-dimensional
model, for which a global analysis of the attractors is a nontrivial problem.

The mechanical model is given by
\begin{equation}
\label{eq:pendulum}
\dot{x}_p = x_v
\qquad \quad 
\dot{x}_v = -\alpha(x_p) - c x_v + u
\end{equation}
where $x_p$ and $x_v$ are position and velocity of the mass respectively,
$u$ is the force input to the system, $c$ is the damping coefficient, and 
$\alpha(x_p) := \partial U(x_p)$ is the force deriving from the mechanical potential
$U:\real \to \real$. 
Contraction, or $0$-dominance is expected with sufficient damping if the potential is strictly convex.
A differential quadratic storage is easily found for the numerical value $c=5$ and  the assumption $$1 \leq \partial \alpha(x_p) \leq 5 \, .$$ 
$P$ is computed via convex  relaxation \eqref{eq:diff-internalLMIrelax} 
for $A_1 := \smallmat{0 & 1 \\ -1 & -5}$ and $A_2 := \smallmat{ 0 & 1 \\ -5 & -5}$.
For    $\lambda = 0$ and  $\varepsilon = 0.01$, the LMI solver (Yalmip \cite{Yalmip2004}, SeDuMi \cite{Sedumi1999} ) returns
$$
P := \mymatrix{cc}
{
 0.8696  &   0.1482 \\
 0.1482  &  0.1304
} 
$$
which is positive definite.

The same approach is repeated for the non-convex potential $$-2 \leq \partial \alpha(x_p) \leq 5 \, ,$$
which allows for several minima (including the classical double-well potential of the nonlinear Duffing 
model \cite[Chapter 2]{guckenheimer1986}).
Figure \ref{fig:locus1} (left) suggests a stable splitting between the two eigenvalues.
The differential storage is computed again by convex relaxation \eqref{eq:diff-internalLMIrelax} for
$A_1 := \smallmat{0 & 1 \\ 2 & -5}$ and $A_2 := \smallmat{ 0 & 1 \\ -5 & -5}$.
For $\lambda =2$ and $\varepsilon = 0.01$, the LMI solver returns  
$$
P = \mymatrix{cc}
{
   -5.1987 &   3.6260 \\
    3.6260  &  6.1987
} 
$$ 
which has inertia $1$. Figure \ref{fig:locus1} (right) shows the non positive level sets of $\delta x^T P \delta x$.
The system is $1$-dominant, meaning that every bounded trajectory asymptotically
converges to some fixed point. 
\begin{figure}[htbp]
\hspace{-5mm}
\includegraphics[width=1.1\columnwidth]{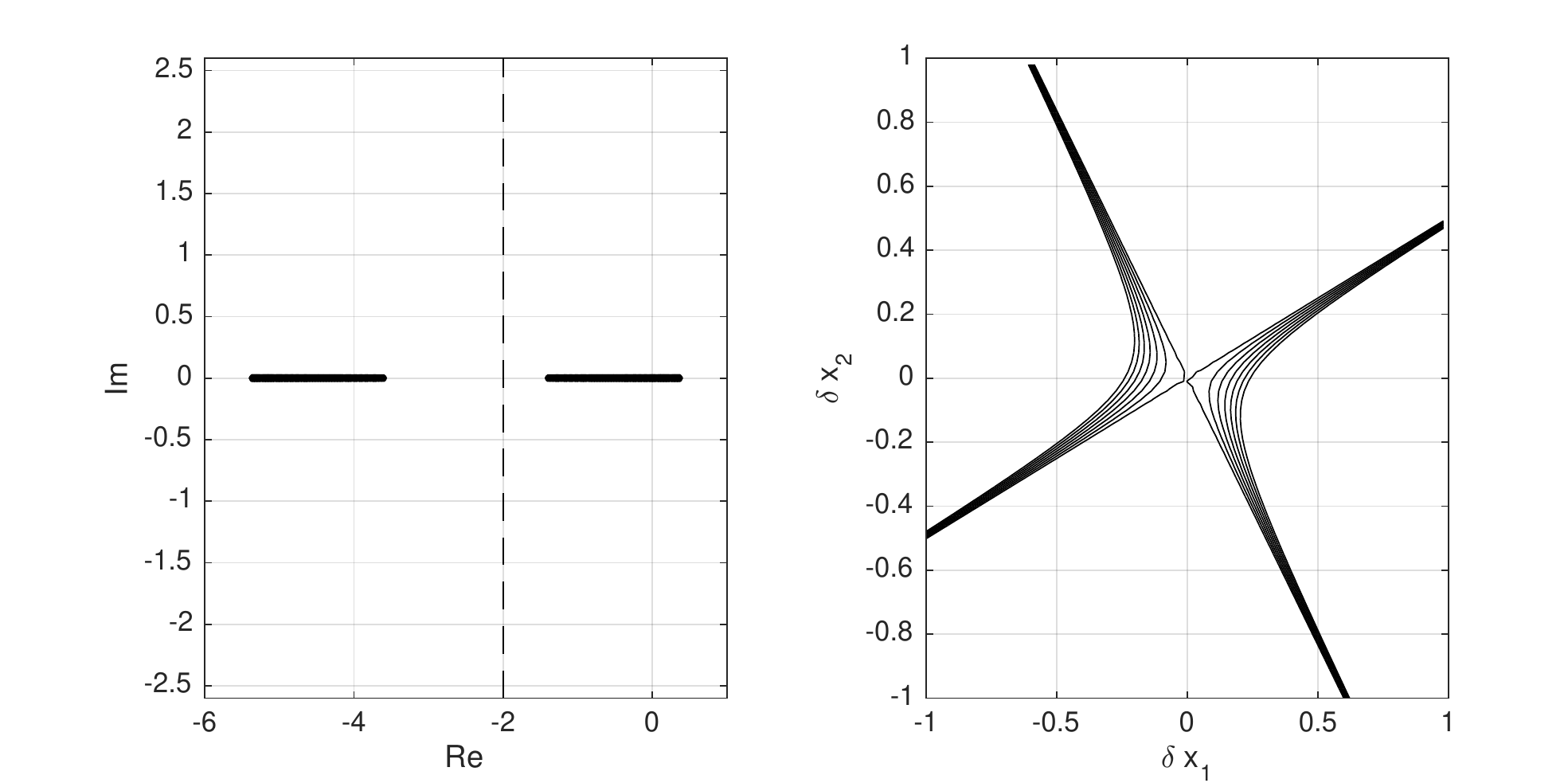}
\caption{\textbf{Left:} roots of the Jacobian of the mechanical systems for different $x_p$. The value
of $\lambda$ is emphasized by the vertical dashed line. 
\textbf{Right:} negative level sets of $\delta x^T P \delta x$ where $P$ has inertia $1$. }
\label{fig:locus1}
\end{figure}

Suppose now that  the mechanical system is driven by a DC motor modelled by the electrical equation
\begin{equation}
\label{eq:DCmotor}
u = k_f x_i 
\qquad \quad
L \dot{x}_i = -R x_i - k_e x_v + V \, .
\end{equation}
$x_i$ is the current of the circuit, $k_f$ is a static approximation of the current to force characteristic, 
$L$ and $R$ are inductance and resistance respectively, 
$k_e$ is the back electromotive force coefficient,
and the voltage $V$ is an additional input.

It is easy to verify that $1$-dominance is preserved for 
$R=1$, $k_f = 1$, $k_e = 1$, and $0 < L< 0.05$, 
since the time-scale separation between 
electrical and mechanical dynamics introduces
a mild perturbation on the dominant/slow dynamics of the system.
For $L=0.1$ the reduced time scale separation
allows for interaction between electrical and
mechanical dynamics. The distribution of the eigenvalues of the Jacobian
in Figure \ref{fig:locus2} (left) suggests that strict $1$-dominance
still holds. Indeed, for $L=0.1$, $\lambda = 2$, and $\varepsilon = 0.01$, the LMI solver
returns 
$$
P = \mymatrix{ccc}
{
   -3.0942  &  0.8985 &  -0.5355 \\
    0.8985  &  3.3771  &  0.1935 \\
   -0.5355  &  0.1935  &  0.7171
} 
$$
which has inertia $1$. For constant inputs $V$, every bounded trajectory necessarily converges to some fixed point,
as illustrated in Figure \ref{fig:locus2} (right), where we considered the double well potential
$U(x_p) := x_p^2/4 + \cos(x_p)$.
\begin{figure}[htbp]
\hspace{-5mm}
\includegraphics[width=1.1\columnwidth]{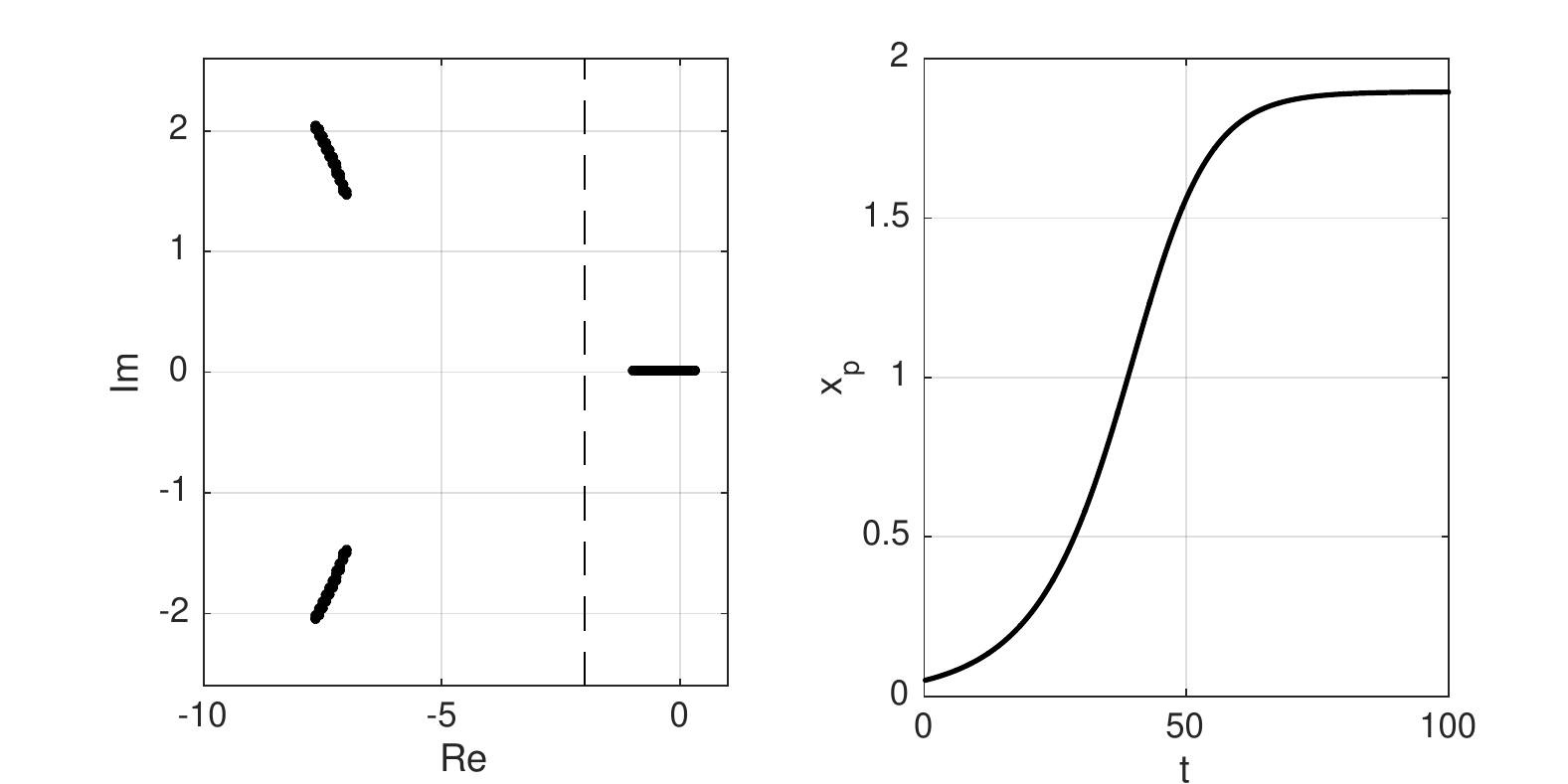}
\caption{\textbf{Left:} roots of the Jacobian of \eqref{eq:pendulum} and \eqref{eq:DCmotor},
by sampling $-2 \leq \partial \alpha(x_p) \leq 5$ at different $x_p$. 
\textbf{Right:} Trajectory of the mass position in time for the interconnected system \eqref{eq:pendulum}, \eqref{eq:DCmotor} 
with potential $U(x_p) := x_p^2/4 + \cos(x_p)$, from the initial condition
$x_p = 0.05$, $x_v = 0$, $x_i = 0$, at constant $V=0$.}
\label{fig:locus2}
\end{figure}

Finally, we close the loop with a PI controller
\begin{equation}
\label{eq:PI}
V = k_P (r-x_p) + k_I x_c \qquad \dot{x}_c = r-x_p 
\end{equation}
where $x_c$ is the integrator variable, $k_P$ and $k_I$ are
proportional and integral gains, respectively, and $r$ is the reference.

The degree of dominance of the closed loop can be modulated via
PI control. A detailed analysis of PI control for $p$-dominance
is beyond the scope of this paper. We just observe that with gains 
$k_P =1$ and $k_I=5$ the eigenvalues of
the Jacobian in  Figure \ref{fig:locus3} (left) exhibits a stable splitting  into two groups of two eigenvalues.
$2$-dominance is verified with $\lambda = 2$, $\varepsilon = 0.01$, in which case the LMI solver returns the storage
$$
P = \mymatrix{cccc}
{
   -4.3713  &  1.7901 &  -0.5507 &  0.0216 \\
    1.7901  &  5.6483  &  0.3768   & -0.9320 \\
   -0.5507  &  0.3768  &  1.0521 &  -0.4363 \\
    0.0216  & -0.9320 &  -0.4363  & -1.3291
} 
$$
with inertia $2$.  For $r=0$ the unique fixed point at $0$ is unstable. We conclude that every bounded
trajectory must converge to a periodic orbit, as illustrated in Figure \ref{fig:locus3} (right).

\begin{figure}[htbp]
\hspace{-5mm}
\includegraphics[width=1.1\columnwidth]{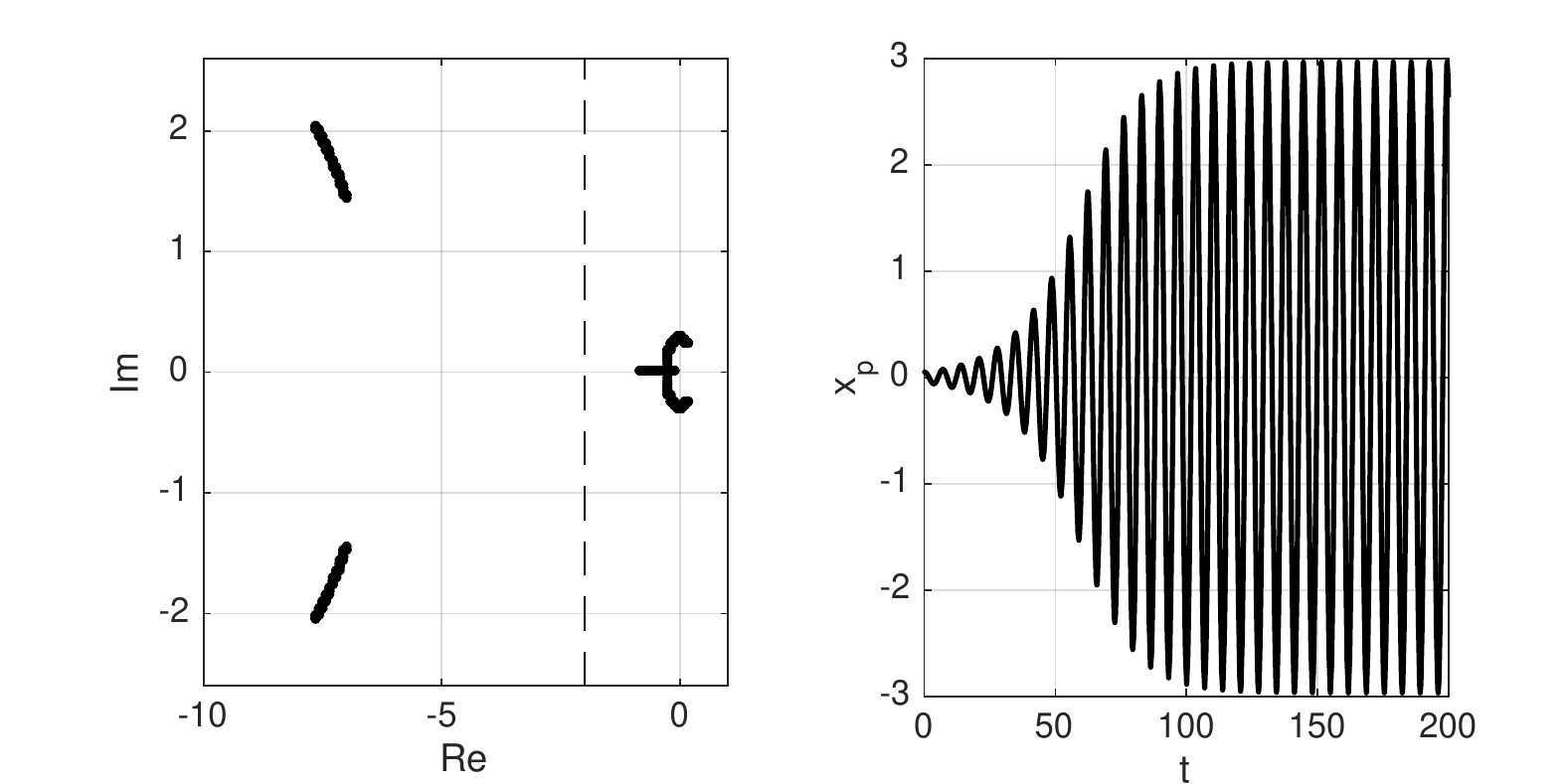}
\caption{\textbf{Left:} roots of the Jacobian of the closed loop given by \eqref{eq:pendulum}, \eqref{eq:DCmotor} and \eqref{eq:PI}
by sampling $-2 \leq \partial \alpha(x_p) \leq 5$ at different $x_p$. 
\textbf{Right:} Trajectory of the mass position in time for the closed loop \eqref{eq:pendulum}, \eqref{eq:DCmotor}, \eqref{eq:PI}
with potential $U(x_p) := x_p^2/4 + \cos(x_p)$ and PI control gains $k_P = 1$ and $k_I = 5$, 
from the initial condition $x_p = 0.05$, $x_v = 0$, $x_i = 0$ and $x_c = 0$, at constant $r=0$.}
\label{fig:locus3}
\end{figure}

\begin{remark}
A variant of the example is to replace the mechanical  model with  the nonlinear pendulum
$\dot x_p = x_v$, $\dot x_v = -\sin(x_p) - c x_v + u$ with
$(x_p,x_v) \in \mathbb{S}\times \real$. One-dominance is proven as in 
Duffing example but the state-space is now nonlinear.  From Theorem \ref{thm:splitting}, 
every attractor of the pendulum exhibits at least one direction of contraction, locally.
The shape of $\mathbb{S}\times \real$ makes $1$-dominance 
compatible with fixed point attractors, for small constant torque $u$,
and with attractors defined by periodic orbits,  for large constant 
torque $u$, \cite[Section 8.5]{Strogatz1994}. See also the analysis in \cite{Forni2016} using differential positivity.
\end{remark}

\section{Interconnections}
\label{sec:interconnections}

\subsection{Differential dissipativity theory}
Dissipativity theory is a fundamental complement to stability theory in systems and control \cite{Willems1972,Willems1972a, Willems2007b}.
Stability, and more generally, dominance, is the property of a {\it closed} system. Dissipativity
is the property of an {\it open system}, i.e. a system with inputs and outputs. By decomposing
a closed system as an interconnection of open subsystems, the search of  quadratic storage
for the analysis of the {\it closed} system is converted into the solution of linear matrix inequalities for 
the {\it open} subsystems. 

The recent papers \cite{Forni2013,Forni2013a,Schaft2013} have proposed to use dissipativity theory {\it differentially} in order to
 study contraction (i.e. $0$-dominance) via interconnections.  We pursue this approach  to study $p$-dominance.  We restrict
  to open systems of the form
\begin{equation}
\label{eq:open-system}
\left\{
\begin{array}{rcl}
\dot{x} &=& f(x) + Bu \\
y &= & C x + D u
\end{array} 
\right. 
\qquad x\in \calX,  (y,u) \in \calW
\end{equation}
where $\calX$ and $\calW$ are smooth manifolds of dimension $n$ and $m$, respectively.
The associated prolonged system reads 
\begin{equation}
\label{eq:open-prolonged}
\left\{
\begin{array}{rcl}
\dot{x} &=& f(x)+Bu \\
\dot{\delta x} &=& \partial f(x) \delta x + B \delta u \\
y &=& C x + Du \\
\delta y &=& C \delta x + D \delta u 
\end{array}
\right.
\end{equation}
where $(x,\delta x) \in T\calX$ and $(w,\delta w) := ((y,u), (\delta y, \delta u)) \in T\calW$.

We also assume that $\calW$ is covered by a single chart 
and that the matrix $L$ and the symmetric matrices $Q$ and $R$ 
below have suitable dimensions. All restrictions above relate to  the restriction of 
{\it constant} tensors $P$ in this paper for the analysis of dominance.

\begin{definition}
\label{def:p-dissipativity}
A nonlinear system \eqref{eq:open-system}
is \emph{differentially $p$-dissipative with rate $\lambda$ 
and differential supply rate}
\begin{equation}
\label{eq:supply}
s(w,\delta w) := \mymatrix{c}{\! \!\!\delta y \!\!\! \\ \!\!\! \delta u\!\!\!}^T \!\!
\mymatrix{cc}{
\!Q\! & \! L \! \\ \! L^T \! & \! R \!
}\!
\mymatrix{c}{\!\!\!\delta y\!\!\! \\ \!\!\! \delta u\!\!\!} 
\end{equation} 
if for some symmetric matrix $P$ with inertia $p$
and some constant $\varepsilon \geq 0$,
the prolonged system
\eqref{eq:open-prolonged} satisfies the conic constraint
\begin{equation}
\label{eq:p-dissipativitydiff}
\mymatrix{c}{\!\!\!\dot{\delta x}\!\!\! \\ \!\!\!\delta x\!\!\!}^T \!\!
\mymatrix{cc}{
\!0\! & \!P\! \\ \!P\! & \!2\lambda P \!+\! \varepsilon I\!
}\!
\mymatrix{c}{\!\!\!\dot{\delta x}\!\!\! \\ \!\!\!\delta x\!\!\!}
\leq
\mymatrix{c}{\! \!\!\delta y \!\!\! \\ \!\!\! \delta u\!\!\!}^T \!\!
\mymatrix{cc}{
\!Q\! & \! L \! \\ \! L^T \! & \! R \!
}\!
\mymatrix{c}{\!\!\!\delta y\!\!\! \\ \!\!\! \delta u\!\!\!} 
\end{equation}
for all $(x,\delta x)\in T\calX$ and 
all $(w, \delta w) \in \calW$.
\eqref{eq:open-system} is \emph{strictly} differentially $p$-dissipative if 
$\varepsilon>0$.
\end{definition}

Differential dissipativity is guaranteed by the feasibility of the inequality
\begin{equation*}
{\footnotesize\mymatrix{cc}
{\!\!
\!\partial\! f(x)^T \!\! P \!\!+\!\! P \partial\! f(x)  \!-\!C^T  \! Q C \!+ \!2\lambda P \!+\! \varepsilon I \!\!\! &\!\! \! P B \!-\! C^T\!L-C^T  \! Q D \!\!\! \\
 \!\! B^T \!P \!-\! L^T C-D^T  \! Q C \!\!  \!\!& \!\!\!\!  \!-\!R\!-\!D^T\!\! L \!-\!L^{\!T}\!\! D \!-\! D^T\!QD\!\!\!
 }}
\!\leq\! 0
\end{equation*}
for some symmetric matrix $P$ with inertia $p$
and some constant $\varepsilon \geq 0$.
A necessary condition for the feasibility of
this inequality is that
$\partial f(x)^T \!P + P \partial f(x) - C^T\! Q C + 2\lambda P + \varepsilon I \leq 0 $
which corresponds to \eqref{eq:diff-internalLMI} for $Q=0$ and clarifies 
the connection between differential dissipativity and dominance.
This infinite family of LMIs can reduced to a finite family through
relaxations, following the approach in Section \ref{sec:convex_relaxation}.

\subsection{A dissipativity theorem for $p$-dominance}

Suppose that a nonlinear system 
can be decomposed as the interconnection of two subsystems:
\begin{subequations}
\label{eq:interconnected}
\begin{align}
\left\{
\begin{array}{rcl}
\dot x_1 &=& f_1(x_1) + B_1 u_1  \\
y_1 &=& \underline{C} x_1 + D_1 u_1
\end{array}
\right.  \label{eq:Sigma1} \\ 
\left\{
\begin{array}{rcl}
\dot x_2 &=& f_2(x_2) + B_2 u_2 \label{eq:Sigma2}\\
y_2 &=& \overline{C} x_2 + D_2 u_2
\end{array}
\right. \\
u = H y + v \hspace{1.5cm} \label{eq:H}
\end{align}
\end{subequations}
where the matrix $H$ specifies the interconnection pattern between
inputs $u = [\, u_1^T \ u_2^T\,]^T$ and outputs $y = [\, y_1^T \ y_2^T\,]^T$.
$v = [\, v_1^T \ v_2^T\,]^T$ is an additional input.
Then the dissipativity theorem below provides conditions for the
differential $p$-dissipativity of the interconnected system from
the differential $p$-dissipativity of its components. 

The formulation of the theorem follows the approach of \cite{Moylan2014}. 
The theorem can be easily adapted to the interconnection of several systems \cite{Moylan1979}. 
In what follows we will use 
$\bar{Q}:= \footnotesize{\mymatrix{c|c}{\!\!\!\!Q_1\!\!\! & \\  \hline \!\! & \!\!\!Q_2\!\!\!\!}} $,
$\bar{L} := \footnotesize{\mymatrix{c|c}{\!\!\!\!L_1\!\!\! & \\  \hline \!\! & \!\!\!L_2\!\!\!\!}}$, and
$\bar{R} := \footnotesize{\mymatrix{c|c}{\!\!\!\!R_1\!\!\! & \\  \hline \!\! & \!\!\!R_2\!\!\!\!}}$
for readability, where for $i \in \{1,2\}$ the matrices $Q_i$, $L_i$ and $R_i$ 
characterize  the differential supply rate fo each subsystem. 

\begin{theorem}
\label{thm:interconnection}
Let \eqref{eq:Sigma1} and \eqref{eq:Sigma2} be 
(strictly) differentially $p_1$-dissipative and 
$p_2$-dissipative
respectively, with
uniform rate $\lambda$ and 
differential supply rate
\begin{equation}
\mymatrix{c}{\!\!\delta y_i\!\! \\ \!\!\delta u_i\!\!}^T \!
\mymatrix{cc}{
Q_i & L_i \\ L_i^T & R_i
}
 \mymatrix{c}{\!\!\delta y_i\!\! \\ \!\!\delta u_i\!\!}
 \qquad i \in \{1,2\} \, .
\end{equation}
Then, the interconnected system  \eqref{eq:interconnected}
is (strictly) differentially $p$-dissipative with degree $p= p_1+p_2$, 
rate $\lambda$, and differential supply rate 
\begin{equation}
\mymatrix{c}{\!\!\!\delta y \!\!\! \\ \!\!\!\delta v\!\!\!}^{\!T} \!\!
\mymatrix{cc}{
Q & L \\ L^T & R
}
 \mymatrix{c}{\!\!\delta y\!\! \\ \!\!\delta v\!\!} 
\end{equation}
where
\begin{equation}
\begin{array}{rcl}
Q &:= & \bar{Q} + \bar{L}H + H^T \bar{L}^T + H^T\bar{R}H \\
L &:=& \bar{L} + H^T\bar{R} \\
R &:=& \bar{R} \ .
\end{array}
\end{equation}
For $v=0$, 
the interconnected system  \eqref{eq:interconnected}
is (strictly) $p$-dominant if $Q \leq 0$. 
\end{theorem}
\begin{proof} 
Let $P_1$ and $P_2$ be solutions to \eqref{eq:p-dissipativity}
respectively for \eqref{eq:Sigma1} and \eqref{eq:Sigma2}. 
Take $P := P_1 + P_2$, $x := [\, x_1^T \ x_2^T\,]^T$, and 
$\delta x := [\, \delta x_1^T \ \delta x_2^T\,]^T$.
Note that $P_1$ has inertia $p_1$ and
 $P_2$ has inertia $p_2$, 
 where $n_1$ and $n_2$ are the dimensions of the
 two state manifolds, respectively. It follows that
 $P$ has inertia $p_1+p_2$.
 Furthermore, by (strict) differential dissipativity of 
 \eqref{eq:Sigma1} and \eqref{eq:Sigma2}, 
 \eqref{eq:p-dissipativity} can be written 
 in the aggregated form
$
\dot{\delta x}^T P \delta x + \delta x^T P \dot{\delta x} 
+ \delta x^T (2\lambda P + \varepsilon I) \delta x 
\leq \delta y^T \overline{Q} \delta y
+ \delta y^T \overline{L} \delta u
+ \delta u^T \overline{L}^T \delta y 
+ \delta u^T \overline{R} \delta u
$,
for some $\varepsilon \geq $ ($\varepsilon > 0$).
Since $\delta u = H \delta y + \delta v$, the expression above reads
$
\dot{\delta x}^T P \delta x + \delta x^T P \dot{\delta x} 
+ \delta x^T (2\lambda P + \varepsilon I) \delta x 
\leq \delta y^T Q \delta y
+ \delta y^T L \delta v 
+  \delta v^T L^T \delta y 
+ \delta u^T R \delta u
$
which shows (strict) differential $p$-dissipativity of the interconnected system.

Finally, for $v=0$ we get
$
\dot{\delta x}^T P \delta x + \delta x^T P \dot{\delta x} 
+ \delta x^T (2\lambda P + \varepsilon I) \delta x 
\leq \delta y^T Q \delta y \leq 0
$
where the last inequality follows from the condition $Q\leq 0$.
This shows that $P$ is a solution of \eqref{eq:diff-internal}.
\end{proof}

Theorem \ref{thm:interconnection} is a general result to study dominance via interconnections.
In particular, it provides an interconnection theorem for the analysis of contraction ($0$-dominance)
and monotonicity ($1$-dominance).

\subsection{Differential passivity analysis}
The passivity supply
\begin{equation}
\label{eq:passivity}
\mymatrix{c}{\! \!\!\delta y \!\!\! \\ \!\!\! \delta u\!\!\!}^T \!\!
\mymatrix{cc}{
\!0\! & \!I \! \\ \! I \! & \! 0 \!
}\!
\mymatrix{c}{\!\!\!\delta y\!\!\! \\ \!\!\! \delta u\!\!\!} 
\end{equation} 
plays an important role in dissipativity theory because it connects the theory to
the {\it physical} property that a passive system can only store the energy supplied by its environment.
The theory of port-Hamiltonian systems encompasses a broad modeling framework of physical models
with passivity properties \cite{vanderschaft2014}.

The passivity theorem states that the feedback interconnection of passive systems is passive. A differential 
version of this important theorem is provided by 
Theorem \ref{thm:interconnection}: any system that is differentially $p$-passive preserves that property under
the feedback interconnection with a differentially passive (i.e. $0$-passive) system.

Passivity theory has proven useful in identifying robust controller structures that preserve stability. An important
particular case is the class of proportional-integral controllers, that we now revisit in the light of dominance analysis.

We first consider the {\it proportional} feedback controller 
\begin{equation}
\label{eq:P}
(P) \; \; \; \bar{y}  = k_P(\bar{u})  \ .
\end{equation}
The controller is differentially $0$-passive (for arbitrary nonnegative rate) 
from $\bar{u}$ to $\bar{y}$ provided that the mapping $k_p(\cdot)$
is monotone, or differentially positive:  $\partial k_P(\bar{u}) \geq 0$ for all $\bar{u}\in \real$.
For instance, 
$$
\mymatrix{c}{\! \!\!\delta \bar{y} \!\!\! \\ \!\!\! \delta \bar{u}\!\!\!}^{\!T} \!\!
\mymatrix{cc}{
\!0\! & \!I\!  \\ \! I \! & \! 0 \!
}\!
\mymatrix{c}{\!\!\!\delta \bar{y}\!\!\! \\ \!\!\! \delta \bar{u}\!\!\!} 
= 2\delta \bar{u}^T \partial k(\bar{u})\delta \bar{u} 
\geq 0.
$$

Likewise, the {\it proportional-integral} feedback controller 
\begin{equation}
\label{eq:P}
(PI) \; \; \; \dot x_c = \bar{u}, \; \bar{y}  = k_P(\bar{u}) + k_I x_c 
\end{equation}
from $\bar{u}$ to $\bar{y}$ is 
\begin{itemize}
\item differentially $0$-passive if $k_P(\cdot)$ is monotone and if $k_I \geq 0$,
with rate $\lambda = 0$. 
\item differentially $1$-passive if $k_P(\cdot)$ is monotone and if $k_I < 0$,
with rate $\lambda \geq 0$ (strictly for $\lambda > 0$).
\end{itemize}
This is because  the storage $S(\delta x_c) := \frac{k_I}{2}\delta x_c^2$ satisfies 
$\dot S \le \delta \bar{u} \delta \bar{y}$. Furthermore,
for $\lambda > 0$ and $k_I < 0$ we have 
$\dot S + \lambda k_I \delta x_c^2 + \varepsilon \delta x_c^2
 \le \delta \bar{u} \delta \bar{y}$
for $\varepsilon = \lambda |k_I|$.

As an illustration, we revisit the nonlinear (Duffing) mass-spring-damper system in Section \ref{sec:examples_dominance}.
For any nonlinear spring satisfying $-3 \leq \partial \alpha(x_p) \leq 3$, 
the system is strictly differential $1$-passive with rate $\lambda \!=\! 2$ from $u$ to $y\!=\!-x_p$: defining the state $x := [\,x_p\ x_v\,]$, the  
variational dynamics $\dot{\delta x} = A(x)\delta x + B \delta u$, $\delta y= C\delta x$
$$
A(x) := \mymatrix{cc}{\!0\! & \!1\! \\ \!-\partial \alpha(x_p)\! & \!-5\!} \quad B := \mymatrix{c}{\!0\! \\ \!1\!} \quad C := \mymatrix{cc}{\!-1\! & \!0\!}.
$$
satisfies
$$
\begin{array}{rcl}
A(x)^T P + PA(x) + 2 \lambda P + \varepsilon I &\leq& 0  \\
P B &=& C^T
\end{array}
$$
for $P = \smallmat{-2 & -1 \\ -1 & 0}$ and $\varepsilon = 0.01$, for all $x\in \real^2$. 

Assume for simplicity that $k_P(\cdot)$ is odd and monotone
and consider the interconnection
$$
u = \bar{y} + v\qquad \bar{u} = -y \ .
$$
The closed loops with the proportional controller
and with the proportional-integral controller are 
illustrated in Figure \ref{fig:interconnection}, for the
the particular case $k_P(\bar u) := \tanh(2 \bar u)$ and $k_I := -1$.

By Theorem \ref{thm:interconnection}, 
the proportional feedback $u=-k_P(y) + v= k_P(x_p)+v$ preserves $1$-differential passivity.
Figure \ref{fig:route_to_oscillations} (left) illustrates a situation where the original system is monostable (linear spring, i.e.  a quadratic potential)
 and then converted to a bistable system with the
 static output feedback $u = -\tanh(2y) + v = \tanh(2 x_p) + v$. The example is elementary but illustrative of a general principle: the feedback controller
 shapes the potential energy of a contractive mechanical  system to convert the system from monostable to bistable.

Likewise,  
the proportional-integral control $u = -k_P(y) - k_I \int y + v = k_P(x_p) - k_I \int y + v = \tanh(2 x_p) - \int x_p + v$
 makes the closed-loop system differentially $2$-passive
because it corresponds to the negative feedback interconnection of two differentially $1$-passive systems.
Figure \ref{fig:route_to_oscillations2} illustrates that, in this new configuration, the role of the proportional controller is to convert a linear stable system ($k_P=0$) into
a $2$-dominant system with a stable limit cycle oscillation ($k_P=1$).

Again the example is elementary but illustrates the general principle that PI control can turn a contractive mechanical system into a system with a limit cycle attractor.
Such conclusions are usually drawn from a local analysis of the linearization around the stable equilibrium of the contractive system. The first scenario corresponds to a saddle node bifurcation whereas the second scenario corresponds to a Hopf bifurcation. The differential approach in this paper makes this analysis non local.

\begin{figure}[htbp]
\centering
\includegraphics[width=0.92\columnwidth]{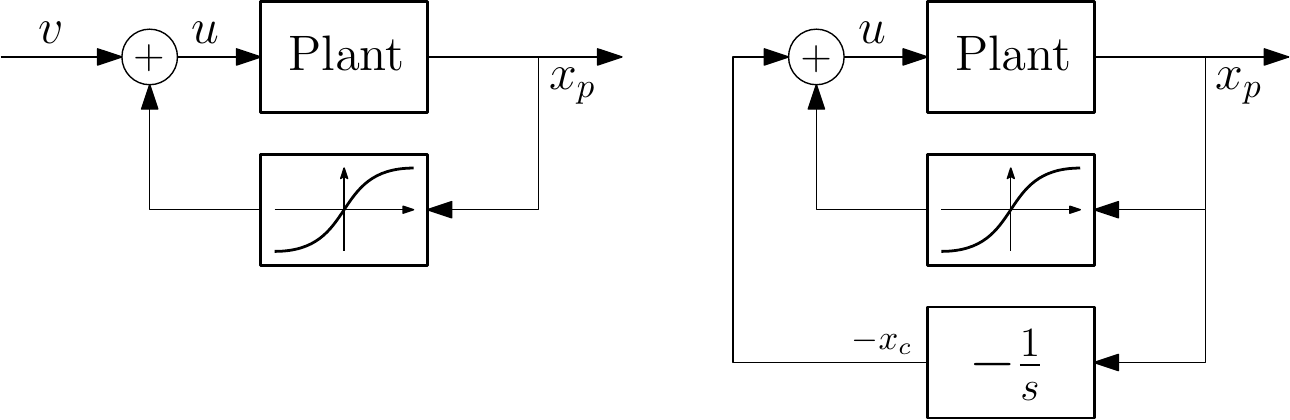}
\caption{
\textbf{Left:} Closed loop of the plant given by \eqref{eq:pendulum} with $\alpha(x_p) = x_p$ and $u=\tanh(2x_p) + v$ (saturated proportional feedback).
\textbf{Right:} Closed loop given by the interconnection of the linear mass-spring-damper system and nonlinear 
proportional-integral controller.}
\label{fig:interconnection}
\end{figure}
\begin{figure}[htbp]
\centering
\includegraphics[width=1\columnwidth]{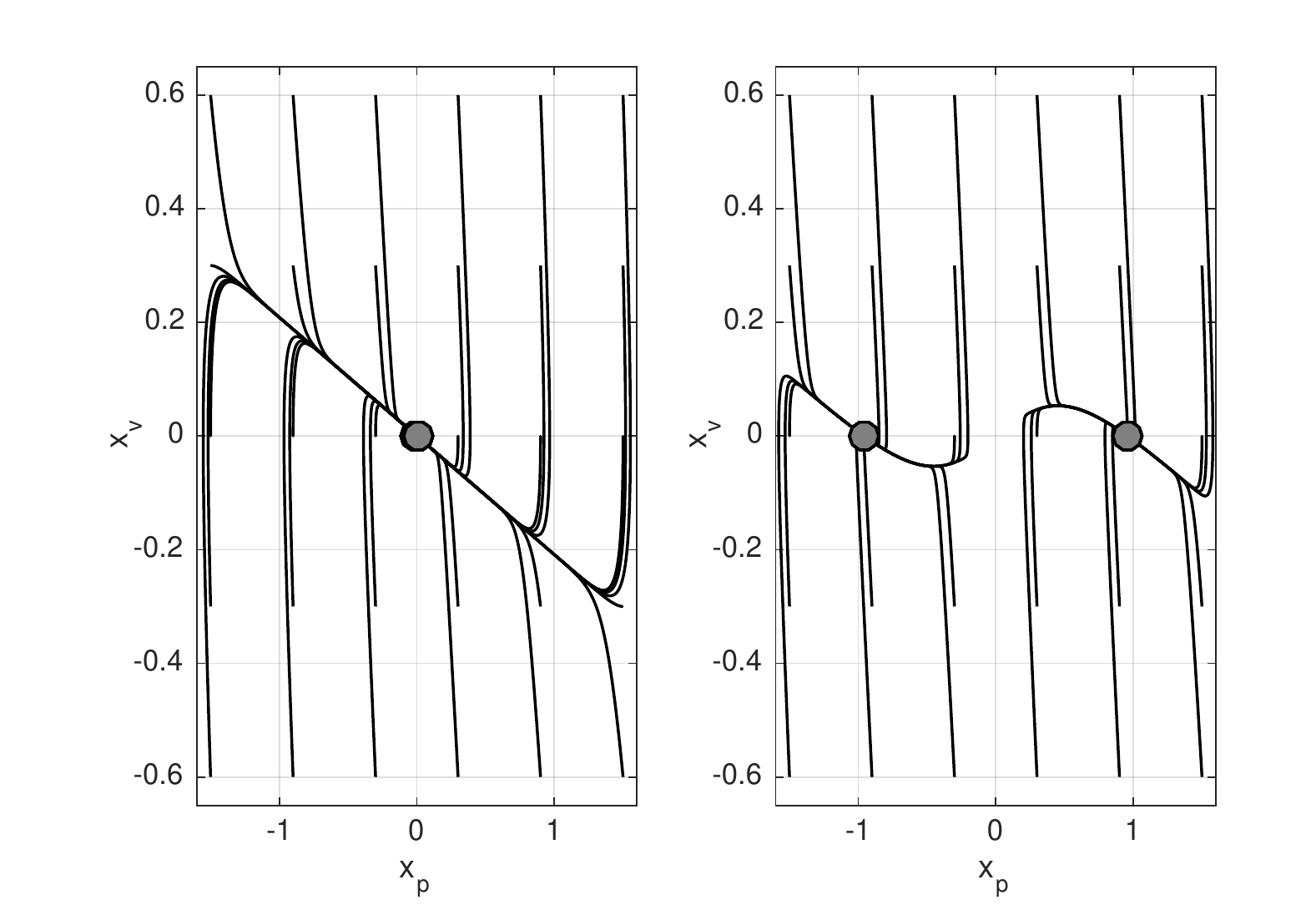}
\caption{
\textbf{Left:} trajectories from different initial conditions of the open loop system \eqref{eq:pendulum} with $\alpha(x_p) = x_p$ and $u=0$. 
\textbf{Right:} trajectories from different initial conditions of the closed loop \eqref{eq:pendulum},$u =  \tanh(2 x_p) + v$, with $\alpha(x_p) = x_p$ and $v=0$.}
\label{fig:route_to_oscillations}
\end{figure}
\begin{figure}[htbp]
\centering
\includegraphics[width=1\columnwidth]{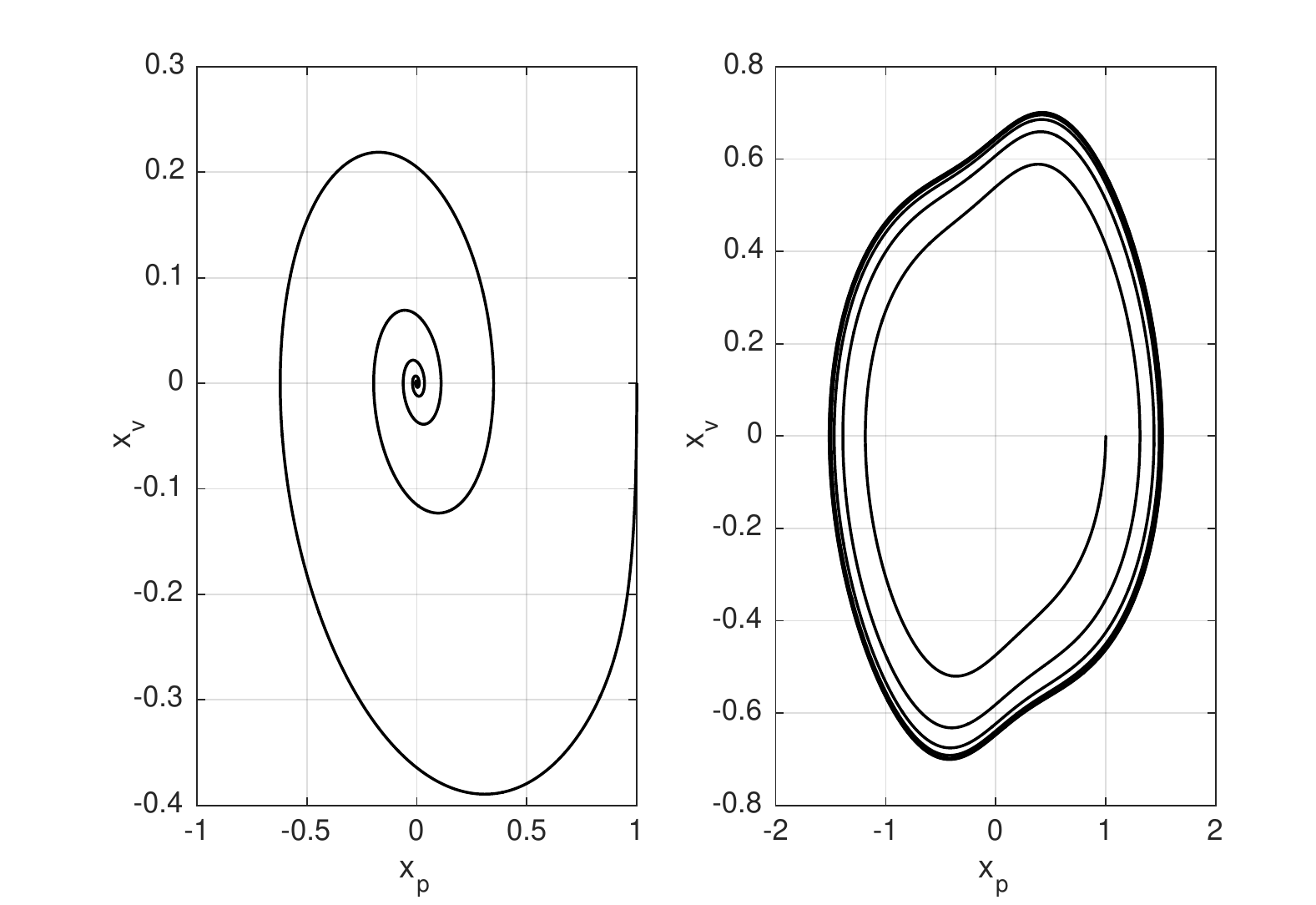}
\caption{
    \textbf{Left:} Trajectory of the linear closed-loop system ($k_P=0)$. 
\textbf{Right:}   Trajectory of the nonlinear closed-loop system $(k_P=1)$ . Parameters: $\alpha(x) = x_p$. Initial condition $x_p = 1$, $x_v = 0$, $x_c = 0$.}  
\label{fig:route_to_oscillations2}
\end{figure}

\subsection{Differential small gain analysis}

The supply
\begin{equation}
\label{eq:small-gain}
\mymatrix{c}{\! \!\!\delta y \!\!\! \\ \!\!\! \delta u\!\!\!}^T \!\!
\mymatrix{cc}{
\!-I\! & \!0 \! \\ \! 0 \! & \! \gamma^2 I \!
}\!
\mymatrix{c}{\!\!\!\delta y\!\!\! \\ \!\!\! \delta u\!\!\!} \ .
\end{equation} 
also has a special status in dissipativity theory because of its connection to the small gain theorem, 
a cornerstone of robust control theory \cite{Zames1966a,Zames1968b,Desoer1975,Zhou1995,VanDerSchaft1999}.

When specialized to the supply (\ref{eq:small-gain}),  Theorem \ref{thm:interconnection} provides a differential version of the small gain theorem:
$p$-dominance is preserved under feedback with a $0$-dominant system provided that their finite differential gains $\gamma_1$ and $\gamma_2$ (of degree $p$ and $0$ respectively)
satisfy the small gain condition $\gamma_1\gamma_2 < 1$.

The differential small gain theorem opens the way to a differential approach of nonlinear robust control.
The link with an operator-theoretic definition of the differential gain of a dominant system is beyond the scope
of this paper (see e.g. \cite{Georgiu1993} for a concept  of differential gain for stable systems), but we briefly illustrate how
the theorem can be used in nonlinear robustness analysis.

\begin{figure}[htbp]
\centering
\includegraphics[width=0.48\columnwidth]{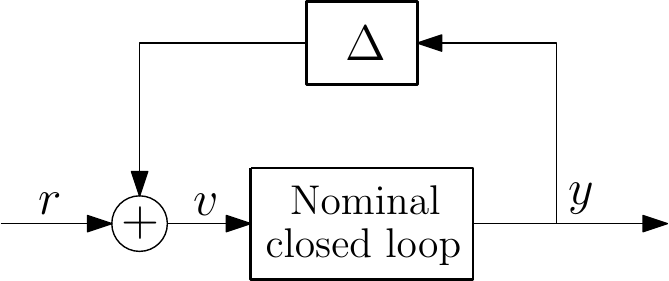}
\caption{
A nonlinear plant with parametric uncertainties represented as the feedback 
interconnection of nominal plant and uncertain dynamics $\Delta$.
}
\label{fig:small_gain_loop}
\end{figure}

Consider the nominal closed loop in Figure \ref{fig:interconnection} (right)
where the mass-spring-damper system \eqref{eq:pendulum}
with linear spring $\alpha(x_p) = x_p$ is interconnected to a saturated proportional-integral
feedback $u = \tanh(2 x_p) - \int x_p + v$. We study the robustness of the oscillations
to perturbations affecting the spring constant and damping coefficient. We use the usual
representation of parametric uncertainties through feedback interconnections 
as shown in Figure \ref{fig:small_gain_loop}.

For $c=5$ and $\lambda = 2$ the nominal closed loop 
system is $2$-dominant with differential gain $\gamma := 0.5636$ from the input $v$ to the position output $x_p$,
obtained for 
$$P := \mymatrix{ccc}{   -0.5522 &   0.0498 &  -0.0171 \\    0.0498  &  1.4946  & 0.3068 \\ -0.0171  &  0.3068  &  0.0576} .$$ 
The nonlinear mechanical spring is modeled through an additive perturbation, i.e.
$$
 \alpha(x_p) = x_p + \Delta(x_p) \ ,
$$
corresponding to the feedback interconnection 
in Figure \ref{fig:small_gain_loop}, with $v = \Delta(x_p)$.
The differential small-gain theorem implies that
$2$-dominance is preserved for any 
$|\partial \Delta (x_p)| \leq 1/\gamma$.
Furthermore, for any perturbation $\Delta(x_p)$ that preserves the origin as a  unique and unstable fixed point of the
closed loop system, every bounded trajectory of the perturbed closed loop will converge to a periodic orbit, like 
in the nominal case. 

The analysis of perturbations affecting the damping coefficient $c=5$ is similar.
For $c=5$ and $\lambda = 2$ the nominal closed loop 
system is $2$-dominant with differential gain $\gamma := 0.5468$ from the input $v$ to the velocity output $x_v$,
obtained for $$P := \mymatrix{ccc}{   -0.2859  & -0.0028 &  -0.0131 \\   -0.0028  &  1.2328 &   0.2977 \\   -0.0131  &  0.2977  &  0.0532}.$$ 
Oscillations will persist when the linear damping coefficient $c x_v$ is
replaced by a nonlinear coefficient $c(x_v)$ given by 
$$
c(x_v) = 5 x_v + \Delta(x_v)
$$
provided that $ |\partial \Delta(x_v) | < 1 / \gamma$ and that the origin remains an unstable fixed point.

\section{Conclusion}

This paper illustrated that linear-quadratic dissipativity theory, a cornerstone of stability theory,
generalizes with surprising ease to the analysis of {\it dominance}. Dominance analysis in turn
is relevant to capture the frequent property that the asymptotic behavior of a nonlinear
dynamical model is low-dimensional. In particular,  the theory seems relevant 
to generalize the theory of stability to a theory of  {\it multistability} and {\it limit cycle} analysis.

The approach in this paper is {\it differential}, meaning that the usual linear matrix inequalities of dissipativity theory
are considered in the tangent bundle. They characterize a dominated splitting of the linearized flow
between $p$ dominant directions and $n-p$ transient directions. This property is captured with
the usual linear matrix inequalities of dissipativity theory, with the only difference that the solution
matrix $P$ is required to have a fixed inertia ($p$ negative eigenvalues and $n-p$ positive eigenvalues),
the standard stability framework corresponding to the case $p=0$.

An important restriction throughout this paper is to analyze $p$-dominance with a constant
quadratic storage (i.e. a constant $P$). This restriction is the price to be paid for tractability.
Standard LMI solvers can then be used to construct the storage.

A number of generalizations deserve further attention. Those include the construction of differential
storages that are non quadratic, or/and state-dependent (i.e. non constant $P(x)$), as well 
as the study of dominance with state-dependent rate $\lambda(x)$, or systems with a different
degree of dominance in different parts of the state-space, or the analysis of dominance in 
non smooth systems. Such generalizations have received considerable attention in the analysis
of contraction, i.e. $0$-dominance, suggesting clear avenues to study $p$-dominance.
Finally, differential dissipativity theory offers an opportunity to revisit classical results
from robust  control theory and absolute stability theory in the context of multistable and
oscillatory systems, see e.g. \cite{Miranda-Villatoro2017} for a first step in that direction.

 \linespread{0.92}

\bibliographystyle{plain}

\end{document}